\documentclass[conference]{IEEEtran}
\IEEEoverridecommandlockouts
\usepackage{amsmath,amssymb,amsthm,mathtools,bm}
\usepackage[T1]{fontenc}
\usepackage{enumitem}
\usepackage{graphicx}
\usepackage{booktabs}
\usepackage{multirow}
\usepackage{algorithm}
\usepackage{xcolor}
\usepackage{algorithmic}
\usepackage[colorlinks=true,linkcolor=blue!60!black,citecolor=blue!60!black,urlcolor=blue!60!black]{hyperref}

\renewcommand{\d}{\mathrm{d}}

\newtheorem{proposition}{Proposition}

\newtheorem{definition}{Definition}
\newtheorem{assumption}{Assumption}
\newtheorem{corollary}{Corollary}

\newcommand{\Ad}{\mathrm{Ad}}
\newcommand{\ad}{\mathrm{ad}}

\newcommand{\Tr}{\mathrm{Tr}}

\newcommand{\R}{\mathbb{R}}
\newcommand{\C}{\mathbb{C}}

\begin{document}
\title{
A Reservoir Computing Approach \\ to Quantum Gate Synthesis
}

\author{
\IEEEauthorblockN{\textsc{Francesco Caravelli}\textsuperscript{1},
\textsc{Roberto Menta}\textsuperscript{2,1},
\textsc{Antonio Sannia}\textsuperscript{3,1,4}} 
\vspace{0.2cm}
\IEEEauthorblockA{
\textsuperscript{1}\textit{Theoretical Division, Los Alamos National Laboratory, Los Alamos, NM 87544, USA} \\
\textsuperscript{2}\textit{NEST, Scuola Normale Superiore, I-56127 Pisa, Italy}\\
\textsuperscript{3}\textit{IFISC (UIB-CSIC), Campus Universitat de les Illes Balears, 07122 Palma de Mallorca, Spain}\\
\textsuperscript{4}\textit{USRA Research Institute for Advanced Computer Science (RIACS), USA}\\ [2pt]
\url{caravelli@lanl.gov}, \url{roberto.menta@sns.it}, \url{sannia@ifisc.uib-csic.es}}
}

\maketitle

\begin{abstract}
Quantum gate synthesis is essential for implementing quantum algorithms
on real hardware, yet existing methods are often computationally
demanding. Here, we introduce a novel approach based on reservoir
computing, which we name Group Reservoir Computing, an efficient machine-learning paradigm for learning
temporal dynamics whose training reduces to a single linear regression,
to reduce the resources
required. The method is grounded in the Wei--Norman decomposition, which provides a compact description of the evolution. We prove that the reconstructed dynamics always remain unitary by construction and derive formal error bounds that establish the theoretical validity of the strategy. On the standard single-qubit gate set the trained network produces a control pulse in a single pass, with mean fidelity $0.94$ across the eight benchmark gates; used to warm-start gradient-based optimization, it roughly halves the number of iterations that plain gradient ascent needs to reach a target fidelity, so that the relevant figure of merit is the time to reach that threshold rather than the final accuracy after a fixed budget. Owing to its general
formulation, the method applies to any finite-dimensional hardware
platform; the route to multiqubit synthesis is discussed in the closing
section.

\end{abstract}

\section{Introduction}
\label{sec:intro}

Every quantum algorithm ultimately runs as a sequence of gates, and
each gate has to be realized physically by driving the hardware with
carefully shaped control signals. What the experimenter has to work
with is partly fixed and partly adjustable: the device comes with an
intrinsic, always-on dynamics set by nature, together with a small
number of adjustable control fields: drive amplitudes, pulse shapes,
magnetic fields, or gate voltages. Quantum gate synthesis is the
problem of choosing the time profiles of these few controls so that the
system, evolving under them, reproduces a desired gate as accurately as
possible~\cite{Koch2022, Lloyd2014, Viola1999, Viola2009}.

Established methods such as \texttt{GRAPE}~\cite{KhanejaGRAPE},
Krotov~\cite{SklarzTannor2002}, and their descendants all require
simulating the full quantum evolution, at a cost of order $O(KN^3)$ per
iteration for $K$ time steps in Hilbert-space dimension $N$. There is a
clear need for faster, cheaper substitutes for the map from controls to
gate. The difficulty is that any such shortcut must still respect the
one property that cannot be given up: a quantum gate must be unitary.

Our approach is based on a classical result from Lie theory, the
Wei--Norman decomposition~\cite{WeiNorman1963,WeiNorman1964} (WN). It shows
that the complicated, matrix-valued evolution of the gate can be tracked
through just a few ordinary numbers, a set of time-dependent
coordinates that obey simple equations driven by the control fields.
This recasts a hard problem, learning an evolution that lives on the
curved space of unitary matrices, into a control problem of a nonlinear dynamical system. Input--output learning of this kind is
exactly what reservoir computing does well. A reservoir, here an
echo-state network~\cite{Jaeger2001,LukoseviciusJaeger2009}, is a
fixed, randomly wired dynamical system whose response to an input is
extracted by a single trained linear readout, so training reduces to one
linear regression rather than the delicate optimization of a deep
network. We let the reservoir learn how the coordinates evolve under
the controls, and then reassemble the gate from them. The essential
point is that this reassembly is a mathematical map that always lands on
the set of valid unitaries: even when the learned coordinates are
imprecise, the gate it returns is genuinely unitary by construction because of the WN decomposition. A learning error
can move the gate within the allowed space, but it can never push it
outside.

This is a fundamentally different use of learning from quantum reservoir
computing~\cite{FujiiNakajima2017,MartinezPena2021,MujaldArranz2021,
GhoshOpala2021, Sannia2024, Venturelli2026}, where a quantum system itself plays the role of the
reservoir, and it is used for learning or generating classical temporal data. Here, by contrast, the coherent gate itself
is the object we keep and reconstruct by using a classical dynamical system. A complementary approach was proposed in Ref.~\cite{Ghosh2021}, where an ancillary quantum network is employed to generate quantum gates on a smaller subsystem. In contrast, our approach leverages classical resources to reduce the required quantum ones. For this reason, it might be more appropriate to call this approach \textit{Group Reservoir Computing}, as it applies both to classical and quantum control problems.

The architecture therefore separates into two clean parts: a reservoir
absorbs the complicated, history-dependent dynamics, while a fixed
mathematical map enforces, exactly, and only at the end, that the
output is a valid gate. This separation is what makes the scheme
attractive for hardware. The reservoir can be built from a fast
analog~\cite{Appeltant2011}, photonic~\cite{Vandoorne2014}, or even
physical quantum~\cite{Fernando2003, menta2025building, aiudi2026} systems, while the gate is
reassembled by a lightweight digital step. Because validity is imposed
only at that last step, imperfections or noise in the reservoir can
lower the fidelity but, by construction, can never produce an unphysical
result. The same trained model can also
serve as an informed starting guess that speeds up conventional
optimizers.

Machine learning has been brought to quantum control in several
forms: reinforcement learning of control policies~\cite{NiuRL2019},
supervised networks that map gates directly to
pulses~\cite{Nikaeen2024}, physics-informed networks that build the
equations of motion into the training objective~\cite{PINN2025}, and
network architectures based on Lie-theoretic (Cartan) decompositions~\cite{KPQNN2025}---while the Wei--Norman decomposition has its
own history in the control engineering of bilinear and two-level
systems~\cite{Nihtila2010}, and in quantum control it has been used to analyze
controllability and to generate sequential unitary gates from the continuous-time
Schr\"odinger equation~\cite{Altafini2002,Altafini2003,Gen-Altafini2002,Altafini2005}. What sets the present work apart is the
object being learned: not the control pulse and not the gate, but the
intermediate Wei--Norman coordinates, from which the gate is then
recovered exactly. To our knowledge, combining the Wei--Norman
reduction with reservoir computing is itself new, and it is this pairing
that yields a surrogate which respects the group structure by
construction, trains with a single regression, and comes with provable
error bounds.

Our contributions are the following. (i)~We formulate this
structure-preserving framework for controlled dynamics on compact matrix
Lie groups and prove bounds that carry the error in the learned
coordinates over to the error in the gate. (ii)~We
introduce an \emph{inverse} reservoir that, given a target gate,
produces a control pulse in a single pass, using a frequency-domain
description of the pulse to resolve the fact that many different pulses
realize the same gate. (iii)~We show that where this one-shot method
fails, it fails for a geometric reason---the failures cluster at the
singular points of the coordinate system, not at any fundamental
obstruction. (iv)~We combine the two: the inverse reservoir supplies the starting
guess for a standard optimizer (\texttt{GRAPE}). With an accelerated,
line-search variant of the optimizer both warm and cold starts converge
within ten iterations on a single qubit; against plain gradient ascent,
whose cold-start convergence takes $10^2$--$10^3$ iterations, the
learned warm start roughly halves the iterations-to-threshold on seven
of the eight standard gates. We argue that this iterations-to-threshold
count, rather than the final accuracy after a fixed
budget~\cite{Caneva2009}, is the right way to judge a learned
initializer, and that its margin should widen with the system
dimension. 

\begin{figure*}
 \centering
 \includegraphics[width=1\linewidth]{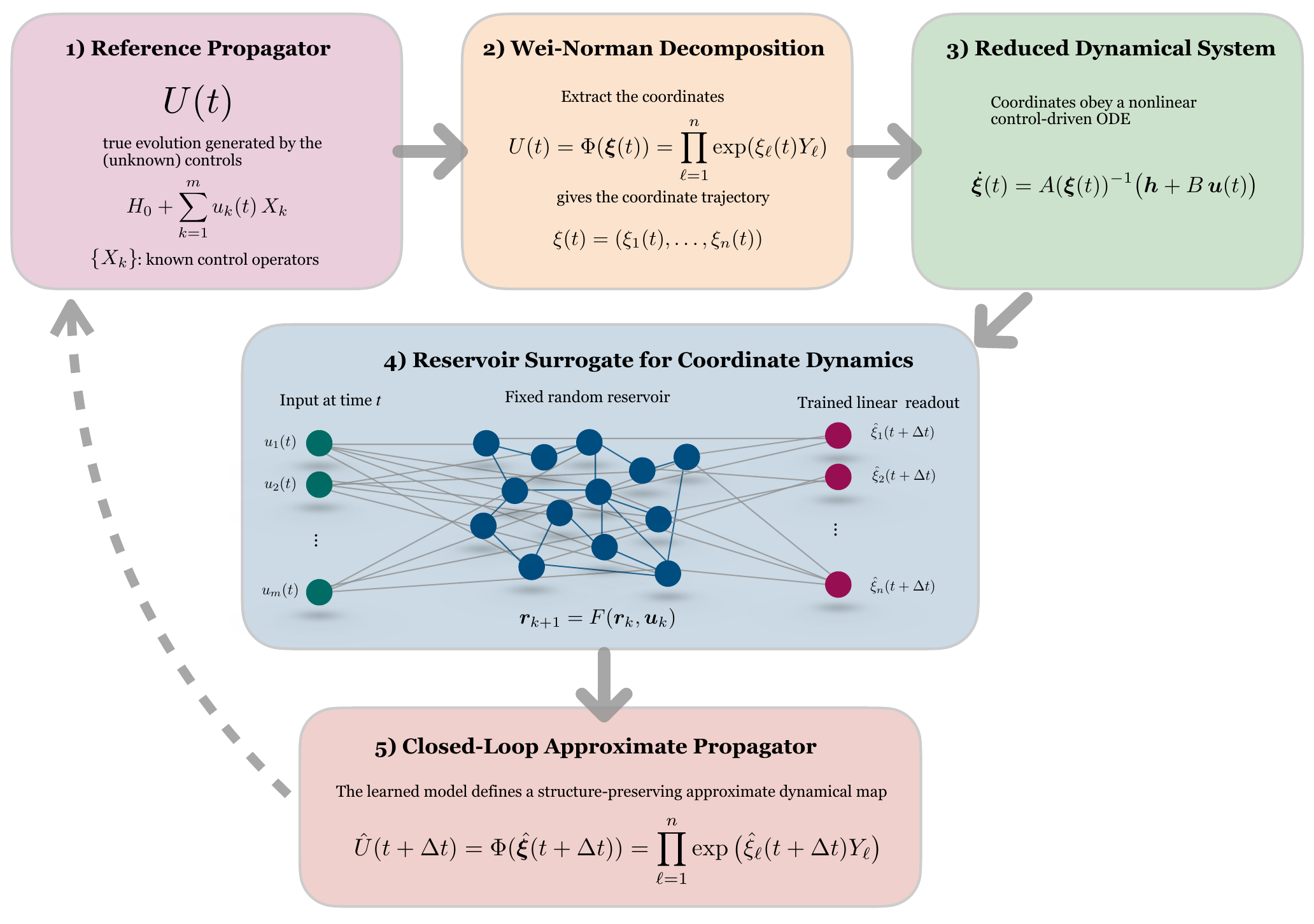}
\caption{Structure-preserving reservoir-computing architecture for
controlled quantum dynamics and gate synthesis. (1) Starting from the
controlled propagator $U(t)$, generated by a Hamiltonian
$H(t)=H_0+\sum_k u_k(t) X_k$, (2) the evolution is expressed in Wei--Norman
coordinates, $U(t)=\Phi(\bm\xi(t))=\prod_{\ell} \exp(\xi_\ell(t)Y_\ell)$. (3) This yields a reduced nonlinear control system for the coordinates,
$\dot{\bm \xi}(t)=A(\bm \xi(t))^{-1}(\bm{h}+B\,\bm{u}(t))$, which is approximated by a
reservoir surrogate driven by the control fields $u_k(t)$. (4,5) The trained
linear readout predicts the coordinate evolution $\hat{\bm\xi}(t)$, and the
propagator is reconstructed by the Wei--Norman product
$\hat U(t)=\Phi(\hat{\bm\xi}(t))$. Because reconstruction is performed
directly on the Lie group, the learned approximation remains
structure-preserving: $\hat U(t)$ belongs to the correct unitary group
even when the surrogate is imperfect.}
\label{fig:fig1}
\end{figure*}

The rest of the paper is organized as follows. Section~\ref{sec:WN}
develops the Wei--Norman reduction, introduces the echo-state surrogate,
and proves the error bounds. Section~\ref{sec:results} presents the
numerical study: we first validate the forward surrogate on the rotation
group $SO(3)$, a transparent classical testbed that shares the same
underlying structure as a single qubit, and then develop the three
single-qubit synthesis strategies (a forward reservoir with gradient
optimization, the single-pass inverse reservoir, and the hybrid
reservoir-plus-\texttt{GRAPE} pipeline) on the Hadamard and other
standard gates. Section~\ref{sec:discussion} closes with extensions and
open problems.

\section{Wei--Norman reduction \\ and reservoir surrogates}
\label{sec:WN}

We start by showing how a controlled evolution
equation on a matrix Lie group can be reduced, via the Wei--Norman
decomposition, to a nonlinear ODE for a finite set of scalar
coordinates. We begin with the necessary Lie-theoretic background,
introduce the Wei--Norman ansatz, derive the resulting coordinate ODE,
and then specialize to the controlled setting that arises in quantum
control.

Let $G$ be a matrix Lie group with Lie algebra
$\mathfrak g\subset\C^{N\times N}$. For any $X,Y\in\mathfrak g$, the Lie bracket
is the matrix commutator $[X,Y]=XY-YX$, whose non-vanishing is the sole
source of difficulty in what follows: were the generators to commute,
the propagator would factor trivially. We recall that the Wei--Norman
factorization introduced below is global for solvable Lie algebras but
only local in the semisimple case relevant here; the coordinate
singularities encountered in Sec.~\ref{sec:WN_SO3_SU2} are the price of
this locality. The exponential map
$\exp:\mathfrak g\to G$ sends a generator $X$ to the finite
transformation $e^X=\sum_{k\ge0}X^k/k!$. Conjugation by a group element
defines the adjoint action $\Ad_g(Y)=gYg^{-1}$, a linear map
$\mathfrak g\to\mathfrak g$; its infinitesimal counterpart is $\ad_X(Y)=[X,Y]$, the derivative of $\Ad_{e^{tX}}$ at $t=0$.
The two are related by $\Ad_{e^X}=e^{\ad_X}$, equivalently the Hadamard
formula
\begin{equation}
 e^X\,Y\,e^{-X}
 = \sum_{k=0}^\infty \frac{1}{k!}\,(\ad_X)^k(Y).
 \label{eq:Hadamard}
\end{equation}

\subsection{Wei--Norman reduction of controlled dynamics}

Now suppose $\dim\mathfrak g=n$ and fix an ordered basis
$(Y_1,\dots,Y_n)$ of $\mathfrak g$. Consider the matrix initial-value
problem:
\begin{equation}
 \dot U(t)=M(t)\,U(t),
 \quad U(0)=I,
 \quad M(t)=\sum_{j=1}^n a_j(t)\,Y_j,
 \label{eq:WN_evolution}
\end{equation}
where the coefficient functions $a_j(t)$ are given. The solution
$U(t)$ lies in $G$ for all~$t$, but computing it directly requires
integrating an $N\times N$ matrix ODE. The Wei--Norman
method~\cite{WeiNorman1963,WeiNorman1964} replaces this operator-valued
problem by a system of $n$ scalar ODEs through the ansatz
\begin{equation}
 U(t)
 = e^{\xi_1(t)Y_1}\,e^{\xi_2(t)Y_2}\cdots e^{\xi_n(t)Y_n},
 \qquad \xi_\ell(0)=0.
 \label{eq:WN_ansatz}
\end{equation}
The ansatz~\eqref{eq:WN_ansatz} is exact but only shifts the task: we
must now find the equations of motion obeyed by the scalar coordinates
$\xi_\ell(t)$. These follow by differentiating the product and matching
against~\eqref{eq:WN_evolution}. The key observation is that, since
$U(t)\in G$, the right logarithmic derivative $\dot U U^{-1}$ is itself
an element of the Lie algebra $\mathfrak g$, and can therefore be
expanded in the basis $(Y_1,\dots,Y_n)$; the following identity makes
that expansion explicit.

\begin{proposition}[Wei--Norman identity {\cite{WeiNorman1963}}]
\label{prop:WN}
For $U(t)$ of the form~\eqref{eq:WN_ansatz},
\begin{equation}
 M(t)
 = \dot U\,U^{-1}
 = \sum_{\ell=1}^n \dot\xi_\ell(t)
 \left(\prod_{k<\ell} e^{\xi_k(t)\,\ad_{Y_k}}\right)Y_\ell.
 \label{eq:WN_identity}
\end{equation}
\end{proposition}

\begin{proof}
Write $U=E_1\cdots E_n$ with $E_\ell=e^{\xi_\ell Y_\ell}$. Since
$Y_\ell$ is constant, $\dot E_\ell=\dot\xi_\ell\,Y_\ell\,E_\ell$, and
differentiating the product gives
$\dot U=\sum_\ell E_1\cdots E_{\ell-1}\,\dot E_\ell\,E_{\ell+1}\cdots E_n$.
Right-multiplying by $U^{-1}=E_n^{-1}\cdots E_1^{-1}$, the factors to the
right of $\dot E_\ell$ cancel and $\dot E_\ell E_\ell^{-1}
=\dot\xi_\ell Y_\ell$, leaving
$\dot U\,U^{-1}=\sum_\ell\dot\xi_\ell\,
\Ad_{E_1\cdots E_{\ell-1}}(Y_\ell)$.
Using $\Ad_{gh}=\Ad_g\circ\Ad_h$ and $\Ad_{e^X}=e^{\ad_X}$ to rewrite
each $\Ad_{E_k}=e^{\xi_k\ad_{Y_k}}$ completes the proof.
\end{proof}

Equation~\eqref{eq:WN_identity} is the heart of the reduction. Each
coordinate velocity $\dot\xi_\ell$ enters through its generator $Y_\ell$
\emph{dressed} by the conjugations $\prod_{k<\ell}e^{\xi_k\ad_{Y_k}}$
induced by all the preceding factors: the first generator ($\ell=1$) is
undressed, while each later generator carries the accumulated action of
the earlier coordinates, evaluated through the Hadamard
series~\eqref{eq:Hadamard}. Were the $Y_\ell$ to commute, every
conjugation would collapse to the identity and the
coordinates would simply integrate the coefficients; all the nontrivial
$\bm\xi$-dependence of the reduction is thus a direct consequence of the
noncommutativity of $\mathfrak g$.

Expanding each conjugated generator in the fixed basis defines the
$n\times n$ Wei--Norman matrix $A(\bm{\xi})$ through
\begin{equation}
 \left(\prod_{k<\ell} e^{\xi_k\,\ad_{Y_k}}\right)Y_\ell
 = \sum_{j=1}^n A_{j\ell}(\bm{\xi})\,Y_j,
\end{equation}
so that $\bm{a}(t) = A(\bm{\xi}(t))\,\dot{\bm{\xi}}(t)$. When
$A(\bm{\xi})$ is invertible, the Wei--Norman ODE is
\begin{equation}
 \dot{\bm{\xi}}(t)
 = A(\bm{\xi}(t))^{-1}\,\bm{a}(t),
 \qquad \bm{\xi}(0)=\bm{0}.
 \label{eq:WN_ODE}
\end{equation}
The right-hand side depends nonlinearly on $\bm{\xi}$ through $A^{-1}$
but only linearly on the driving signal $\bm{a}(t)$. 

We now specialize to the controlled setting. A controlled quantum
system evolves as
\begin{equation}
 \dot U(t)
 = -\frac{i}{\hbar}\Big(H_0+\sum_{k=1}^m u_k(t)\,X_k\Big)U(t),
 \qquad U(0)=I,
 \label{eq:schrodinger}
\end{equation}
where $H_0$ is the drift Hamiltonian, the $u_k(t)$ are control
fields and $X_k$ the corresponding control operators; we set $\hbar=1$
henceforth. Writing $M_0:=-iH_0$ and $M_k:=-iX_k$,
the generator in~\eqref{eq:WN_evolution} is affine in the controls,
\begin{equation}
 M(t)=M_0+\sum_{k=1}^m u_k(t)\,M_k ,
 \label{eq:M_affine}
\end{equation}
and every term lies in $\mathfrak g$. Expanding the drift and each
control generator in the fixed basis $(Y_1,\dots,Y_n)$ defines a constant
offset vector $\bm{h}\in\R^n$ and a constant matrix $B\in\R^{n\times m}$,
\begin{equation}
 M_0=\sum_{j=1}^n h_j\,Y_j,
 \qquad
 M_k=\sum_{j=1}^n B_{jk}\,Y_j ,
 \label{eq:hB_def}
\end{equation}
i.e.\ $\bm{h}$ collects the coordinates of the drift and the $k$-th
column of $B$ collects the coordinates of the $k$-th control generator;
both are time-independent because $H_0$ and the $X_k$ are. Substituting
\eqref{eq:hB_def} into \eqref{eq:M_affine} and collecting the coefficient
of each $Y_j$ gives the coordinate vector of $M(t)$,
\begin{equation}
 a_j(t)=h_j+\sum_{k=1}^m B_{jk}\,u_k(t),
 \label{eq:a_decomp}
\end{equation}
i.e.~$\bm{a}(t)=\bm{h}+B\,\bm{u}(t)$, 
so the drift enters as the constant offset $\bm{h}$ and the controls
enter linearly through $B$. The Wei--Norman ODE then takes the form
\begin{equation}
 \dot{\bm{\xi}}(t)
 = A(\bm{\xi}(t))^{-1}\big(\bm{h}+B\,\bm{u}(t)\big)
 =: f(\bm{\xi}(t),\bm{u}(t)).
 \label{eq:xi_control}
\end{equation}

\subsection{Reservoir surrogates}
\label{sec:reservoir}

Reservoir computing provides a general framework for learning
input--output maps generated by driven dynamical systems. The idea is
to replace the direct construction of a nonlinear model by a two-layer
architecture: a high-dimensional \emph{reservoir} with fixed internal
dynamics, and a simple trained \emph{readout} that extracts the desired
output from the reservoir state. The reservoir may be recurrent,
nonlinear, and random, but it is not trained; only the readout is
optimized. This makes reservoir computing attractive when the
underlying dynamics are complicated but the target map is causal and
regular enough to be captured through a rich state embedding.

In the present setting, the object of interest is not the group-valued
propagator $U(t)$ directly, but the Wei--Norman coordinate trajectory
$\bm{\xi}(t)$ associated with it. Once the propagator has been reduced
to the coordinate dynamics
\begin{equation}
\dot{\bm{\xi}}(t)=f(\bm{\xi}(t),\bm{u}(t)),
\end{equation}
the task becomes learning the control-to-coordinate map. More
precisely, the Wei--Norman equation~\eqref{eq:xi_control} induces a
causal input--output operator $\mathcal{H}:\bm{u}|_{[0,t]}\longmapsto \bm{\xi}(t)$, associating to each admissible control history its coordinate state at
time $t$. We restrict attention to a fixed finite horizon $[0,T]$, with
bounded controls and trajectories confined to a compact set on which
$A(\bm{\xi})$ is invertible; there $\mathcal{H}$ is continuous, and we approximate it empirically
by the leaky-integrator ESN introduced below, with the
resulting surrogate error controlled explicitly in
Proposition~\ref{prop:coord_error}. 

Among reservoir models, ESNs are especially convenient: the nonlinear
recurrence encodes the recent input history, while training reduces to
fitting a linear readout. In our case the readout target is the
Wei--Norman coordinate trajectory itself, and the gate is reconstructed
only after the coordinate prediction has been made. We use a
leaky-integrator ESN with reservoir state $\bm{r}_k = (r_{k,1},\ldots , r_{k,R})\in\mathbb{R}^R$
updated according to
\begin{align}
 &\bm{r}_{k+1}
 = (1-\alpha)\,\bm{r}_k
 + \alpha\,\tanh\!\bigl(W_r \bm{r}_k
 + W_{\mathrm{in}} \bm{u}_k + \bm{b}\bigr), \nonumber \\
 &\bm{r}_0=\bm{0}.
 \label{eq:ESN_update}
\end{align}
Here $W_r$ is a sparse random recurrent matrix scaled to have spectral
radius $\rho<1$, a standard prescription for the echo-state property
(for the leaky-integrator update the relevant contraction condition
involves the effective matrix $(1-\alpha)I+\alpha W_r$),
$W_{\mathrm{in}}$ is the input coupling matrix,
$\bm{b}$ is a bias, and $\alpha\in(0,1]$ is the leak rate, which sets
the effective memory timescale of the network. 

To improve expressivity while keeping the readout linear in its
parameters, we form an augmented state
\begin{equation}
\bm{s}_k=(\bm{r}_k,\bm{r}_k^{\circ 2}\big|_{R_q},\bm{u}_k,1),
\end{equation}
where $\bm{r}_k^{\circ 2}\big|_{R_q} := (r_{k,1}^2,\ldots , r_{k,R_q}^2)$, $R_q\le R$, denotes the componentwise square of the
first $R_q$ reservoir components ($R_q=100$ in all experiments below). The readout then predicts the Wei--Norman coordinates
through
\begin{equation}
 \hat{\bm{\xi}}_k
 = \bm{\xi}_0 + W_{\mathrm{out}}\,\bm{s}_k,
 \qquad
 \hat U_k = \Phi(\hat{\bm{\xi}}_k)\in G ,
 \label{eq:readout}
\end{equation}
where $W_{\mathrm{out}}\in\mathbb{R}^{n\times(R+R_q+m+1)}$ is the trained
readout matrix and $\bm{\xi}_0$ is the (known) initial coordinate. Here
$\bm{\xi}_0=\bm{0}$ by Eq.~\eqref{eq:WN_ansatz}.\footnote{In the
numerical experiments of Sec.~\ref{sec:results} the Wei--Norman
integration is instead initialized at $\bm{\xi}_0=(0,\pi/2,0)$, inside
the chart and away from the singularity at $\xi_2=0$. The realized
physical propagator is then $\Phi(\bm{\xi}(T))\,\Phi(\bm{\xi}_0)^{-1}$,
and a physical target $U_f$ is supplied to the synthesis pipeline as the
ZYZ coordinates of $U_f\,\Phi(\bm{\xi}_0)$.}
Thus the ESN produces a surrogate coordinate trajectory
$\hat{\bm{\xi}}_k$, and the approximate propagator is obtained by
reconstruction on the group; the map $\Phi$ enforces the group
constraint exactly, regardless of the quality of the coordinate
approximation. 

\begin{definition}[Structure-preserving surrogate]
\label{def:SP}
A surrogate is said to be \emph{structure-preserving at the
reconstruction level} if
\begin{equation}
\hat U_k=\Phi(\hat{\bm{\xi}}_k)\in G
\end{equation}
for all time indices $k$ and for all trained parameter values.
\end{definition}

The readout matrix $W_{\mathrm{out}}$ is fitted by ridge regression.
Let $S$ collect the augmented states $\bm{s}_k$ over the training set
and $\Xi$ the corresponding target coordinate displacements obtained
from numerical integration of~\eqref{eq:WN_ODE}. Then
\begin{equation}
W_{\mathrm{out}}
= \Xi S^\top(SS^\top+\lambda I)^{-1},
\end{equation}
with ridge parameter $\lambda>0$. Training thus reduces to a single
linear solve, with all nonlinear processing delegated to the fixed
reservoir dynamics.

We now state the assumptions under which the surrogate error can be
controlled.

\begin{assumption}
\label{ass:main}
Fix a time interval $[0,T]$ and a class of admissible controls
$\mathcal{V}$.
\begin{enumerate}
\item[(i)] The vector field $f(\bm{\xi},\bm{u})$ is Lipschitz in
$\bm{\xi}$ with constant $L$, uniformly over $\mathcal{V}$.
\item[(ii)] Both the exact and approximate trajectories remain in a
compact set $K$ on which $A(\bm{\xi})$ is invertible.
\item[(iii)] The surrogate vector field satisfies
\begin{equation}
\sup_{K\times\mathcal{V}}\|f-\hat f\|\le \varepsilon.
\end{equation}
\end{enumerate}
\end{assumption}

Here $\hat f$ denotes the continuous-time interpolation of the trained
discrete-time surrogate, with the time-discretization error of the
integrator absorbed into $\varepsilon$.

Under these assumptions, the coordinate error satisfies a standard
stability estimate.

\begin{proposition}[Coordinate error]
\label{prop:coord_error}
Under Assumption~\ref{ass:main},
\begin{equation}
\sup_{t\in[0,T]}
\|\bm{\xi}(t)-\hat{\bm{\xi}}(t)\|
\le
e^{LT}
\bigl(
\|\bm{\xi}(0)-\hat{\bm{\xi}}(0)\|
+\varepsilon T
\bigr).
\end{equation}
\end{proposition}

\begin{proof}
Let $e(t)=\bm{\xi}(t)-\hat{\bm{\xi}}(t)$. The exact and approximate
trajectories satisfy
\begin{equation}
\dot{\bm{\xi}}(t)=f(\bm{\xi}(t),\bm{u}(t)),
\qquad
\dot{\hat{\bm{\xi}}}(t)=\hat f(\hat{\bm{\xi}}(t),\bm{u}(t)).
\end{equation}
Subtracting these and adding and subtracting
$f(\hat{\bm{\xi}}(t),\bm{u}(t))$ gives
\begin{align}
\dot e(t)
&=
\bigl[f(\bm{\xi}(t),\bm{u}(t))-f(\hat{\bm{\xi}}(t),\bm{u}(t))\bigr] \nonumber
\\ &+
\bigl[f(\hat{\bm{\xi}}(t),\bm{u}(t))-\hat f(\hat{\bm{\xi}}(t),\bm{u}(t))\bigr].
\end{align}
Writing $e(t)=e(0)+\int_0^t \dot e(s)\,\d s$, taking norms, and using
Assumption~\ref{ass:main} gives
\begin{equation}
\|e(t)\|
\le
\|e(0)\|
+
\int_0^t \bigl(L\|e(s)\|+\varepsilon\bigr)\,\d s.
\end{equation}
Gronwall's inequality then implies
\begin{equation}
\|e(t)\|
\le
e^{Lt}\bigl(\|e(0)\|+\varepsilon t\bigr)
\le
e^{LT}\bigl(\|e(0)\|+\varepsilon T\bigr).
\end{equation}
Taking the supremum over $t\in[0,T]$ proves the claim.
\end{proof}

The coordinate error transfers to the propagator because, for compact
$G$ with anti-Hermitian (or antisymmetric) generators, the
reconstruction map $\Phi$ is globally Lipschitz\footnote{Here
$\|\cdot\|$ denotes the Euclidean ($\ell^2$) norm on vectors and the
Frobenius norm $\|A\|=\sqrt{\Tr(A^\dagger A)}$ on matrices; the latter
is just the Euclidean norm of the matrix entries, so the two are the
same $\ell^2$ norm and we use a single symbol. The constant $L_\Phi$ is
the Lipschitz constant of the reconstruction map $\Phi$, i.e.\ the
smallest number with
$\|\Phi(\bm{\xi})-\Phi(\bm{\eta})\|\le L_\Phi\,\|\bm{\xi}-\bm{\eta}\|$
for all $\bm{\xi},\bm{\eta}$; it measures how strongly an error in the
coordinates is amplified into an error in the reconstructed
propagator.}: each factor $e^{\xi_\ell Y_\ell}$ is unitary, so a change
in coordinates produces a controlled change in $U$.

\begin{corollary}[Propagator error]
\label{cor:prop_error}
Let $G$ be compact with anti-Hermitian generators $Y_\ell$. Then the
reconstruction map $\Phi$ is Lipschitz with
\begin{equation}
L_\Phi \le \sqrt{n}\,\max_{\ell}\|Y_\ell\|,
\end{equation}
and therefore
\begin{equation}
\sup_t \|U(t)-\hat U(t)\|
\le
L_\Phi\, e^{LT}
\bigl(
\|\bm{\xi}(0)-\hat{\bm{\xi}}(0)\|
+\varepsilon T
\bigr).
\end{equation}
For $\mathfrak{su}(2)$ with the basis $Y_\ell=\tfrac{i}{2}\sigma_{a_\ell}$
used below, $\|Y_\ell\|=1/\sqrt{2}$ and hence
$L_\Phi\le \sqrt{3/2}\approx 1.22$.
\end{corollary}

\begin{proof}
Write $\Phi(\bm{\xi})=e^{\xi_1Y_1}\cdots e^{\xi_nY_n}$. For two
coordinate vectors $\bm{\xi},\bm{\eta}$, insert and subtract
intermediate products in which one factor at a time is changed from
$e^{\eta_\ell Y_\ell}$ to $e^{\xi_\ell Y_\ell}$. Since multiplication on
the left and right by unitary matrices preserves the Frobenius norm,
\begin{equation}
\|\Phi(\bm{\xi})-\Phi(\bm{\eta})\|
\le
\sum_{\ell=1}^n
\|e^{\xi_\ell Y_\ell}-e^{\eta_\ell Y_\ell}\|.
\end{equation}
Applying the fundamental theorem of calculus to
$s\mapsto e^{(\eta_\ell+s(\xi_\ell-\eta_\ell))Y_\ell}$ gives
$\|e^{\xi_\ell Y_\ell}-e^{\eta_\ell Y_\ell}\|\le
|\xi_\ell-\eta_\ell|\,\|Y_\ell\|$, so that
\begin{align}
\|\Phi(\bm{\xi})-\Phi(\bm{\eta})\|
\le
\sum_{\ell=1}^n |\xi_\ell-\eta_\ell|\,\|Y_\ell\| \nonumber \\
\le
\sqrt{n}\,\max_\ell \|Y_\ell\|\,\|\bm{\xi}-\bm{\eta}\|.
\end{align}
Applying this with $U(t)=\Phi(\bm{\xi}(t))$ and
$\hat U(t)=\Phi(\hat{\bm{\xi}}(t))$ and invoking
Proposition~\ref{prop:coord_error} gives the stated bound.
\end{proof}

For $SU(2)$ the bound also controls the figure of merit used in
Sec.~\ref{sec:results}: with $\mathcal F(U,V)=|\Tr(U^\dagger V)|^2/4$
one has $1-\mathcal F(U,\hat U)\le \|U-\hat U\|^2/2$, so the propagator
error bounds the gate infidelity directly.

Together, Proposition~\ref{prop:coord_error} and
Corollary~\ref{cor:prop_error} make the division of labor between the two methods clear: while the
reservoir needs only to approximate the reduced coordinate dynamics, the exact group membership is supplied analytically by the Wei--Norman map. 

\subsection{Wei--Norman coordinates for $SO(3)$ and $SU(2)$, and the single-qubit control problem}
\label{sec:WN_SO3_SU2}

We now specialize the general construction to the two three-dimensional
groups most relevant here: the rotation group $SO(3)$ and the
single-qubit unitary group $SU(2)$. Their Lie algebras are isomorphic,
so they share the same coordinate structure at the Wei--Norman level;
$SO(3)$ serves as a classical, geometric prototype, while $SU(2)$ is the
physically relevant setting for one-qubit gate synthesis.

Let $Y_1,Y_2,Y_3$ be an ordered set of generators of the Lie algebra
and write the Wei--Norman factorization
\begin{equation}
U(t)=e^{\xi_1(t)Y_1}e^{\xi_2(t)Y_2}e^{\xi_3(t)Y_3}.
\label{eq:WN_factorization_special}
\end{equation}
By the construction of Sec.~\ref{sec:WN}, the coordinates obey
$\dot{\bm{\xi}}=A(\bm{\xi})^{-1}(\bm h+B\,\bm u)$, with the controls entering
through the affine vector $\bm a=\bm h+B\bm u$.

\paragraph*{The $SO(3)$ prototype}
For $SO(3)$, we adopt the standard ZYZ factorization
\begin{equation}
R(t)=e^{\phi(t)J_z}e^{\theta(t)J_y}e^{\psi(t)J_z},
\label{eq:SO3_ZYZ}
\end{equation}
where $J_x,J_y,J_z\in\mathfrak{so}(3)$ generate rotations about the
coordinate axes. In this chart the Wei--Norman coordinates are the
Euler angles $\bm{\xi}(t)=(\phi(t),\theta(t),\psi(t))^\top$, and the
Wei--Norman matrix is
\begin{align}
&A(\phi,\theta)=
\begin{pmatrix}
0 & -\sin\phi & \cos\phi\,\sin\theta \\
0 & \cos\phi & \sin\phi\,\sin\theta \\
1 & 0 & \cos\theta
\end{pmatrix}, \\
&\det A(\phi,\theta)=-\sin\theta.
\label{eq:A_SO3}
\end{align}
The chart is therefore singular at $\theta=0,\pi$, the familiar
gimbal-lock singularity of ZYZ Euler angles. Away from this set the
reduced dynamics form a well-defined three-dimensional nonlinear control
system, and the $SO(3)$ example serves as a classical testbed with the
same Wei--Norman structure as the single qubit but an intuitive
interpretation in terms of rigid rotations. For later use, solving
$A(\phi,\theta)\dot{\bm{\xi}}=a=(a_1,a_2,a_3)^\top$ gives the explicit
inverse
\begin{align}
\dot\phi &= a_3-\cot\theta\bigl(\cos\phi\,a_1+\sin\phi\,a_2\bigr),
 \label{eq:phi_dot_SO3}\\
\dot\theta &= -\sin\phi\,a_1+\cos\phi\,a_2,
 \label{eq:theta_dot_SO3}\\
\dot\psi &= \csc\theta\bigl(\cos\phi\,a_1+\sin\phi\,a_2\bigr).
 \label{eq:psi_dot_SO3}
\end{align}
Even when the coefficients $a_j(t)$ depend linearly on the controls, the
coordinate dynamics are nonlinear, through the trigonometric dependence
introduced by the inversion of $A(\phi,\theta)$.

\paragraph*{The $SU(2)$ single-qubit case}
We now pass to one-qubit gate synthesis. The single-qubit propagator
evolves in $SU(2)$, and we use the ZYZ factorization
\begin{equation}
U(t)=e^{i\xi_1(t)\sigma_z/2}\,
 e^{i\xi_2(t)\sigma_y/2}\,
 e^{i\xi_3(t)\sigma_z/2},
\label{eq:SU2_ZYZ}
\end{equation}
where $\sigma_x,\sigma_y,\sigma_z$ are the Pauli matrices. The
generators, ordered to match the ZYZ product, are
\begin{equation}
Y_1=\tfrac{i}{2}\sigma_z,\qquad
Y_2=\tfrac{i}{2}\sigma_y,\qquad
Y_3=\tfrac{i}{2}\sigma_z.
\end{equation}
The factorization generators need not be linearly independent---here
$Y_1=Y_3$---since the ordered set is only required to generate the
algebra~\cite{WeiNorman1964}; the matrix $A(\bm{\xi})$ is defined by
expanding in the basis
$\{\tfrac{i}{2}\sigma_x,\tfrac{i}{2}\sigma_y,\tfrac{i}{2}\sigma_z\}$ as
in Sec.~\ref{sec:WN}.
These are anti-Hermitian, as are the $\mathfrak{so}(3)$ generators
$J_a$; the two differ only by a sign convention in the bracket,
$[Y_a,Y_b]=-\varepsilon_{abc}Y_c$ here versus
$[J_a,J_b]=+\varepsilon_{abc}J_c$ for $\mathfrak{so}(3)$. The
Wei--Norman matrix is therefore the $SO(3)$ matrix of
Eq.~\eqref{eq:A_SO3} with its off-diagonal signs flipped (and
$\phi\mapsto\xi_1$, $\theta\mapsto\xi_2$):
\begin{align}
&A(\xi_1,\xi_2)=
\begin{pmatrix}
0 & \sin\xi_1 & -\cos\xi_1\sin\xi_2 \\
0 & \cos\xi_1 & \sin\xi_1\sin\xi_2 \\
1 & 0 & \cos\xi_2
\end{pmatrix}, \\
& \det A(\xi_1,\xi_2)=\sin\xi_2.
\label{eq:A_SU2}
\end{align}

The same coordinate singularity appears at $\xi_2=0,\pi$; it is a
singularity of the chosen chart, not of the gate dynamics. The
corresponding inverse relations are
\begin{align}
\dot\xi_1 &= a_3+\cot\xi_2\bigl(\cos\xi_1\,a_1-\sin\xi_1\,a_2\bigr),
 \label{eq:xi1_dot_SU2}\\
\dot\xi_2 &= \sin\xi_1\,a_1+\cos\xi_1\,a_2,
 \label{eq:xi2_dot_SU2}\\
\dot\xi_3 &= \csc\xi_2\bigl(\sin\xi_1\,a_2-\cos\xi_1\,a_1\bigr).
 \label{eq:xi3_dot_SU2}
\end{align}
These are the single-qubit analogue of the Euler-angle kinematics, but
the reconstruction map~\eqref{eq:SU2_ZYZ} now returns a genuine quantum
propagator in $SU(2)$.

\paragraph*{Analog one-qubit control}
Since our interest is analog gate synthesis for a single qubit, we take
the standard drift-plus-control Hamiltonian
\begin{equation}
H(t)=H_0+\sum_{k=1}^{m}u_k(t)\,X_k,
\label{eq:H_single_qubit_general}
\end{equation}
and, following the numerical experiments,
\begin{equation}
H(t)
=
\tfrac{1}{2}\sigma_z
+
\tfrac{1}{2}\bigl(
u_x(t)\sigma_x+u_y(t)\sigma_y+u_z(t)\sigma_z
\bigr),
\label{eq:H_single_qubit_xyz}
\end{equation}
so that $\dot U(t)=-iH(t)\,U(t)$, $U(0)=I$. Absorbing the
factor $-i$ into the basis and coefficients gives a control vector
$\bm{a}(t)=\bm{h}+B\,\bm{u}(t)$ with $\bm{u}(t)=(u_x,u_y,u_z)^\top$;
for this Hamiltonian and the generators above, $\bm{h}=(0,0,-1)^\top$
and $B=-I_3$, i.e.\
$\bm{a}(t)=-\bigl(u_x(t),\,u_y(t),\,1+u_z(t)\bigr)^\top$. The reduced
Wei--Norman dynamics are then
\begin{equation}
\dot{\bm{\xi}}(t)
=
A(\xi_1(t),\xi_2(t))^{-1}\bigl(\bm{h}+B\,\bm{u}(t)\bigr).
\label{eq:xi_control_SU2_final}
\end{equation}
This is the equation we analyze in the sequel: a three-dimensional
nonlinear control system for $\bm{\xi}(t)=(\xi_1,\xi_2,\xi_3)^\top$,
driven by the analog control fields. Once $\hat{\bm{\xi}}(t)$ has been
predicted by the surrogate, the gate is recovered through the exact
reconstruction
\begin{equation}
\hat U(t)
=
e^{i\hat\xi_1(t)\sigma_z/2}\,
e^{i\hat\xi_2(t)\sigma_y/2}\,
e^{i\hat\xi_3(t)\sigma_z/2}.
\label{eq:Uhat_reconstruction}
\end{equation}
This is the problem we target: given a desired gate $U_f\in SU(2)$, find
analog pulses $u_x(t),u_y(t),u_z(t)$ such that the solution
of~\eqref{eq:xi_control_SU2_final}, reconstructed
through~\eqref{eq:Uhat_reconstruction}, approximates $U_f$ with high
fidelity.

\section{Numerical results}
\label{sec:results}

We validate the framework in two stages. First, we test the forward
surrogate---the ESN that learns the control-to-coordinate map---on
$SO(3)$ Euler-angle dynamics, where the ground truth is transparent and
the hyperparameters can be tuned in isolation. Second, we turn to
single-qubit gate synthesis on $SU(2)$ and develop three strategies:
gradient optimization through the forward surrogate, a single-pass
inverse ESN, and a hybrid that uses the inverse ESN to warm-start
\texttt{GRAPE}. Code for all experiments is available (see Data
Availability). 

\subsection{Forward surrogate validation on $SO(3)$}
\label{sec:forward}

The forward ESN of Sec.~\ref{sec:reservoir} is trained to predict the
Euler-angle trajectory $\bm{\xi}(t)=(\phi,\theta,\psi)$ of
Eq.~\eqref{eq:SO3_ZYZ} from the driving signal, with targets generated
by direct integration of the Wei--Norman ODE. Each driving
signal is a sum of two sinusoids per axis, with frequencies drawn
uniformly from $[0.5,2.0]$, amplitudes drawn uniformly from $[0.2,0.8]$
and scaled by an overall factor $0.8$, and random phases; initial
conditions are sampled near $(0,\pi/2,0)$ with small uniform jitter, and
$\theta$ is kept within $[0.2,\pi-0.2]$ by clipping the coordinate at
each integration step. We sweep the reservoir
hyperparameters around a baseline ($R=1000$, $\rho=0.95$, $\alpha=0.1$,
$\sigma_{\mathrm{in}}=1.0$, $\lambda=10^{-8}$, 150 training
trajectories); the full sweep data are collected in
Appendix~\ref{app:sweep} (Table~\ref{tab:sweeps}). Two findings stand
out. First, the input scaling $\sigma_{\mathrm{in}}$ is by far the most
sensitive parameter (Fig.~\ref{fig:sweep_is}): reducing it to $0.1$
raises the mean reconstruction fidelity to $0.999$, while
$\sigma_{\mathrm{in}}=3.0$ degrades it to $0.833$. Second, small leak
rates ($\alpha\approx 0.01$) and moderate reservoir sizes outperform
larger, more reactive configurations; we attribute the degradation at
large $R$ to the fixed ridge parameter $\lambda$, which
under-regularizes the correspondingly larger readout problem. All sweep
metrics refer to a single reservoir realization evaluated on a single
held-out trajectory, which also accounts for the residual
non-monotonicity in Table~\ref{tab:sweeps}. Both observations point to
the same conclusion: away from chart singularities the Wei--Norman dynamics are
only mildly nonlinear, so the reservoir performs best in the near-linear
regime of its $\tanh$ activation. Reconstruction fidelities, defined for rotations as
$\mathcal F(R,\hat R)=[\Tr(R^\top\hat R)+1]/4$, above $0.99$ on the
held-out trajectory are obtained throughout the well-tuned region.

\subsection{Gate synthesis on $SU(2)$}
\label{sec:synthesis}

We now turn from forward prediction to the inverse problem: given a
target gate $U_f\in SU(2)$, find control pulses $\bm{u}(t)$ that realize
it. We develop three strategies of increasing sophistication and
report their performance on the Hadamard gate and other elements of the
standard universal gate set.

The Hadamard $U_{\mathrm H}=\frac{1}{\sqrt{2}}\bigl(\begin{smallmatrix}1&1\\1&-1
\end{smallmatrix}\bigr)$ has ZYZ decomposition
$\bm{\xi}_f=(\pi,\pi/2,0)^\top$ with fidelity $\mathcal{F}=1$ to machine
precision (up to a global phase, $\Phi(\bm{\xi}_f)=iU_{\mathrm H}$). The physical
Hamiltonian is $H(t)=\tfrac{1}{2}\sigma_z+\tfrac{1}{2}[u_x(t)\sigma_x
+u_y(t)\sigma_y+u_z(t)\sigma_z]$, with $K=100$ time steps at
$\Delta t=0.02$. 

\subsubsection*{Strategy 1: Forward ESN with gradient optimization}
The forward ESN is trained on $SU(2)$ dynamics driven by random Fourier
pulses whose coefficients are drawn with standard deviation
$\mathrm{amp}/\sqrt{n_{\mathrm{modes}}+1}$, $\mathrm{amp}=1.8$, with the
resulting pulses clipped to $|u_d|\le 5$, so as to cover coordinate
displacements of order $\pi$. Control pulses are then
optimized by gradient descent on
$\mathcal{L}=\|\hat{\bm{\xi}}(T)-\bm{\xi}_f\|^2$, with gradients
computed by analytic backpropagation through the fixed reservoir
recurrence. The loss is chart-based and does not identify coordinates
related by the periodic and ZYZ redundancies of the parameterization, so
some equivalent minima are penalized; optimizing $1-\mathcal F$ directly
through the reconstruction is a natural refinement that we leave to
future work. From DC-biased sinusoidal initializations the fidelity
climbs from $\mathcal{F}\approx 0.22$ to $0.91$ by iteration~80 and
reaches $0.993$ by iteration~460, improving slowly thereafter; the
dominant bottleneck is the $\xi_2$ coordinate, which must stay near
$\pi/2$ while $\xi_1$ traverses a distance of $\pi$. Optimization takes
minutes on a laptop, with the ESN providing structure-preserving output
at every iteration. This route is correct but slow, which motivates the
single-pass inverse map below.

\subsubsection*{Strategy 2: Inverse ESN (single-pass synthesis)}
Rather than optimizing through the forward ESN, one can train a separate
\emph{inverse} ESN to map target gate coordinates directly to control
pulses. The difficulty is that the map
$\bm{\xi}_f\mapsto\bm{u}(\cdot)$ is one-to-many: many pulse shapes
produce the same gate, and an ESN trained to reproduce time-domain pulse
samples averages over them and fails. We resolve this with a
\emph{Fourier pulse parameterization}: each control channel is a
truncated Fourier series
\begin{equation}
 u_d(t) = c_{d,0} + \sum_{k=1}^{n_{\mathrm{modes}}}
 \big[c_{d,k}\cos(2\pi k t/T) + s_{d,k}\sin(2\pi k t/T)\big],
 \label{eq:fourier_pulse}
\end{equation}

giving $3(2n_{\mathrm{modes}}+1)$ parameters per pulse ($51$ for
$n_{\mathrm{modes}}=8$). Training pairs are generated by sampling random
Fourier coefficients, integrating the Wei--Norman ODE to obtain
$\bm{\xi}(T)$, and storing $(\bm{\xi}(T),\text{Fourier params})$.
Because the map from Fourier coefficients to pulses is injective, the
only residual degeneracy is the physical one---distinct pulses realizing
the same gate---which the training distribution regularizes but does not
remove. 

Throughout, synthesized gates are evaluated physically. The
Wei--Norman integration is initialized at $\bm{\xi}_0=(0,\pi/2,0)$,
inside the chart, so the propagator realized by a pulse is
$U(T)=\Phi(\bm{\xi}(T))\,\Phi(\bm{\xi}_0)^{-1}$; accordingly, a
physical target $U_f$ is supplied to the network as the ZYZ coordinates
of $U_f\,\Phi(\bm{\xi}_0)$, and fidelities are computed through
$\mathcal{F}(U,V)=|\Tr(U^\dagger V)|^2/4$ between $U_f$ and the realized
propagator. We verified that this reconstruction agrees with direct
integration of the Schr\"odinger equation under
Eq.~\eqref{eq:H_single_qubit_xyz} to better than $10^{-8}$ on random
pulses, so all fidelities quoted below are physical. The full
single-pass pipeline is therefore
\begin{eqnarray}
&\bm{\xi}_{\mathrm{target}}
\longmapsto \text{ESN state}
\longmapsto \bm{\theta}_{\mathrm{F}}
\longmapsto \bm{u}(t)\nonumber \\
&\hspace{1cm}
\longmapsto \bm{\xi}(T)
\longmapsto \Phi(\bm{\xi}(T))\,\Phi(\bm{\xi}_0)^{-1},
\label{eq:pipeline}
\end{eqnarray}
where $\bm{\theta}_{\mathrm{F}}$ collects the Fourier coefficients of the
three channels. The inverse ESN is driven by
$[\bm{\xi}_{\mathrm{target}},\,k/K_{\mathrm{res}}]$ for $K_{\mathrm{res}}=80$ reservoir steps, and its
final-step readout predicts all $51$ coefficients; the run parameters are
listed in Table~\ref{tab:summary}. Note that the input to the inverse
reservoir---constant target coordinates plus a time ramp---differs in
nature from the forward driving signals, so its input scale was set
independently of the $SO(3)$ sweep. The $8000$ training pairs use a
deliberately mixed distribution---$30\%$ near standard gates
(Gaussian coordinate perturbations of the benchmark ZYZ coordinates with
width $0.3$), $30\%$ Haar-distributed ZYZ targets ($\xi_1,\xi_3\sim
U[-\pi,\pi)$, $\cos\xi_2\sim U[-1,1]$, with $\xi_2$ subsequently clipped
to $[0.15,\pi-0.15]$ to avoid the chart singularity), and $40\%$
unconstrained random Fourier pulses---so that
the network sees both the benchmark gates and the broader geometry of
$SU(2)$.

One further ingredient is needed at synthesis time. The ZYZ chart is
multivalued: shifting $\xi_1$ or $\xi_3$ by $2\pi$ changes the
reconstructed gate only by a global sign, so a physical target admits
several coordinate representatives, and the network generalizes well
only for representatives inside the densely sampled region of the
training distribution. We therefore evaluate the network on all
$2\pi$-wrapped representatives of the target (at most nine candidates,
each verified by a millisecond-scale integration of the reduced ODE)
and keep the best. The selection matters: for the $X$ gate it raises
the single-pass fidelity from $0.010$, obtained at the representative
$(-2\pi,\pi/2,-\pi)$ returned by the numerical decomposition, to
$0.9992$ at the equivalent representative $(0,\pi/2,-\pi)$. 

On $300$ held-out test gates the inverse ESN reaches mean fidelity
$0.943$ and median fidelity $0.994$, with $83\%$ of gates above
$\mathcal{F}=0.95$ (full breakdown in
Table~\ref{tab:app_general_test}). The distribution is sharply peaked
near unit fidelity with a thin low-fidelity tail
(Fig.~\ref{fig:inv_fids}): the principal challenge is not typical-case
accuracy but rare, large failures. A representative single-pass
synthesis is shown in Fig.~\ref{fig:inv_traj}---the Hadamard trajectory
approaches the target coordinates and the fidelity grows monotonically,
though the gate is not reproduced exactly before refinement.

On the eight standard gates the single pass attains mean fidelity
$0.938$, with five of the eight above $0.99$
(Table~\ref{tab:hybrid}). The weakest cases are a direct signature of
the coordinate singularity rather than a failure of learning. The
reference frame $\Phi(\bm{\xi}_0)$ maps each physical gate to the
coordinates of $U_f\,\Phi(\bm{\xi}_0)$ and thereby displaces the
singular set across the gate family: in this frame $Z$ and $S$---the
problematic targets of the bare chart---land on the well-conditioned
equator $\xi_2=\pi/2$ and are synthesized at fidelity $0.992$ and
$0.999$, while the Hadamard and $\sqrt{Y}$ are pushed to the chart
boundary $\xi_2=\pi-0.15$ and drop to $0.92$ and $0.79$. The failure
follows the chart position of the target, not the identity of the gate.
For the Hadamard specifically, the single pass yields coordinates
$(-0.75,\,2.64,\,-4.17)$ against the target $(0,\,\pi-0.15,\,-\pi)$,
i.e.\ fidelity $0.921$, in well under a second including the
representative search. This is already a useful warm start: the
limitation is geometric, not algorithmic, and we address it next.

\begin{table}[h]
\centering
\caption{Configuration of the inverse-ESN synthesis experiment.}
\label{tab:summary}
\smallskip
\begin{tabular}{lc}
\toprule
Quantity & Value \\
\midrule
Time steps / pulse & $K=100$ \\
Time step & $\Delta t=0.02$ \\
Pulse duration & $T=2.0$ \\
Fourier modes / channel & $8$ \\
Fourier coefficients / pulse & $51$ \\
Training / test pairs & $8000$ / $300$ \\
Reservoir size & $R=1500$ \\
Spectral radius & $\rho=0.9$ \\
Leak rate & $\alpha=0.03$ \\
Input scale & $\sigma_{\mathrm{in}}=1.0$ \\
Ridge parameter & $\lambda=10^{-7}$ \\
Quadratic readout features & $100$ \\
Reservoir drive length & $K_{\mathrm{res}}=80$ \\
WN initial coordinate & $\bm{\xi}_0=(0,\pi/2,0)$ \\
\bottomrule
\end{tabular}
\end{table}

\begin{figure}[!t]
\centering
\includegraphics[width=\columnwidth]{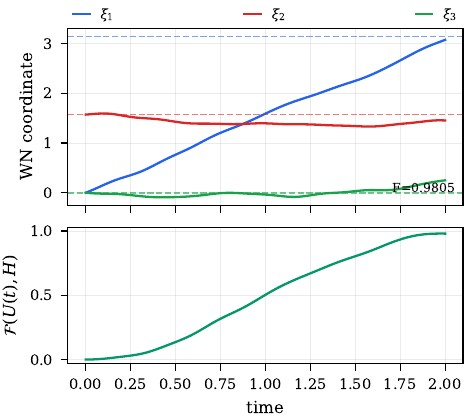}
\caption{Representative single-pass inverse-ESN synthesis of the
Hadamard gate. Top: the Wei--Norman coordinate trajectory approaches
the target coordinates but does not terminate exactly on them. Bottom:
the gate fidelity along the trajectory, confirming that the gate is
synthesized well---though not perfectly---before any gradient
refinement.}
\label{fig:inv_traj}
\end{figure}

\begin{figure}[!t]
\centering
\includegraphics[width=\columnwidth]{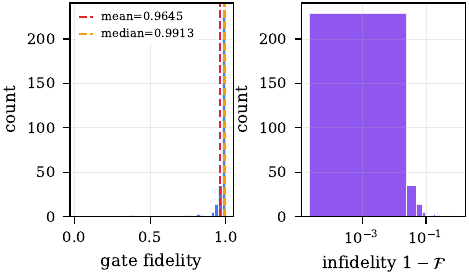}
\caption{Held-out fidelity distribution for the inverse ESN on $300$
unseen gates. The sharp peak near unit fidelity gives the high median
($0.994$), while the long low-fidelity tail accounts for the lower
mean ($0.943$); the principal challenge is the elimination of rare
but substantial failures.}
\label{fig:inv_fids}
\end{figure}

\subsubsection*{Strategy 3: Hybrid (inverse ESN warm start + \texttt{GRAPE})}
The single-pass pulses are refined with
\texttt{GRAPE}~\cite{KhanejaGRAPE}, which computes the exact gradient of
the gate fidelity by forward/backward propagation of the $2\times 2$
Schr\"odinger equation, at cost $O(KN^3)$ per iteration. We consider
two variants of the optimizer: \emph{plain} gradient ascent with a fixed
step, and an \emph{accelerated} variant that normalizes the gradient and
performs a bidirectional backtracking line search at each iteration.
Both stop as soon as $\mathcal F>0.99999$. For every gate we compare
the warm start supplied by the inverse ESN with a cold start from a
random pulse, under the same kernel, budget, and stopping rule; results
are collected in Table~\ref{tab:hybrid} and Fig.~\ref{fig:hybrid}.

Two regimes emerge. Under the accelerated kernel the single-qubit
landscape is benign enough that \emph{both} starts converge within ten
iterations (warm: $2$--$8$, median $5$; cold: $4$--$10$, median $6.5$),
with final fidelities above $0.99999$ in every case: at this system size
a strong optimizer leaves little room for a learned initializer, and
fixed-budget final fidelity is not a useful discriminator. The plain
kernel exposes the mechanism the warm start is meant to exploit:
cold-start convergence now requires $7\times10^2$--$1.8\times10^3$
iterations, and the ESN warm start reduces the iterations-to-threshold
on seven of the eight gates, lowering the median from
$\approx 1.1\times10^3$ to $\approx 5.3\times10^2$; the single exception
is $\sqrt{Y}$, whose warm start is the poorest
($\mathcal F_{\mathrm{ESN}}=0.79$, at the displaced singularity). The
meaningful comparison is therefore \emph{time to a fidelity threshold}
rather than final accuracy after a fixed budget: at one qubit the margin
is a factor of about two against plain gradient ascent and negligible
against the accelerated kernel, and it should widen with the
Hilbert-space dimension, where each iteration is costly and cold starts
are far from any basin (Sec.~\ref{sec:discussion}). The inverse ESN is
thus best viewed not as a replacement for \texttt{GRAPE}, but as a fast
learned initializer that turns a global search into a short local
polish.

In this implementation the ESN stage, including the representative
search, takes $0.55$--$0.91$\,s per gate (mean $0.68$\,s), and the
accelerated warm-started \texttt{GRAPE} stage $36$--$195$\,ms (mean
$127$\,ms), so the full pipeline remains of order one second per
gate.\footnote{Timings refer to a pure NumPy implementation on a
laptop-class CPU (i.e. a personal MacBook).}

\begin{table*}[t]
\centering
\caption{Standard-gate synthesis with the inverse ESN and its hybrid
refinement. \\ Iteration counts are to the stopping threshold $\mathcal
F>0.99999$ from the ESN warm start and from a random cold start, for the
accelerated (line-search) and plain (fixed-step) \texttt{GRAPE} kernels.}
\label{tab:hybrid}
\footnotesize
\setlength{\tabcolsep}{8pt}
\renewcommand{\arraystretch}{1.1}
\begin{tabular}{lcccccc}
\toprule
Gate & $\mathcal{F}_{\mathrm{ESN}}$
 & $\mathcal{F}_{\mathrm{hybrid}}$
 & $N_{\mathrm{warm}}$
 & $N_{\mathrm{cold}}$
 & $N^{\mathrm{plain}}_{\mathrm{warm}}$
 & $N^{\mathrm{plain}}_{\mathrm{cold}}$ \\
\midrule
Hadamard & 0.9205 & 1.000000 & 6 & 7 & 755 & 1012 \\
$X$ (NOT) & 0.9992 & 1.000000 & 4 & 9 & 366 & 1772 \\
$Y$ & 0.8078 & 0.999994 & 5 & 10 & 1413 & 1789 \\
$Z$ & 0.9921 & 0.999999 & 6 & 6 & 557 & 886 \\
$S$ (phase) & 0.9986 & 1.000000 & 2 & 7 & 413 & 702 \\
$T$ ($\pi/8$) & 0.9982 & 0.999999 & 5 & 4 & 435 & 896 \\
$\sqrt{X}$ & 0.9954 & 0.999997 & 3 & 5 & 512 & 1102 \\
$\sqrt{Y}$ & 0.7907 & 0.999995 & 8 & 5 & 2098 & 1102 \\
\bottomrule
\end{tabular}

\medskip
\begin{minipage}{\textwidth}
\footnotesize $N_{\mathrm{warm}}$ and $N_{\mathrm{cold}}$ refer to the
accelerated kernel (budget $100$), $N^{\mathrm{plain}}_{\mathrm{warm}}$
and $N^{\mathrm{plain}}_{\mathrm{cold}}$ to plain gradient ascent
(budget $3000$); all runs use the same per-gate random cold start.
$\mathcal{F}_{\mathrm{hybrid}}$ is the final fidelity of the
warm-started accelerated run; entries displayed as $1.000000$ correspond
to infidelities below $10^{-6}$. The warm start lowers the plain-kernel
median from $\approx 1057$ to $\approx 535$ iterations, winning on seven
of the eight gates; under the accelerated kernel both starts converge
within ten iterations, so the fair comparison is the time to first cross
a chosen fidelity threshold, not the final fidelity.
\end{minipage}
\end{table*}

The ZYZ singularity that limits the pure inverse ESN is a coordinate
artifact, not a physical obstruction. Beyond switching charts (e.g.\
ZXZ when $\xi_2$ is near $0$ or $\pi$), a two-chart atlas, or a
quaternion-based reconstruction, the reference frame $\Phi(\bm{\xi}_0)$
itself offers a handle: since the frame displaces the singular set, a
per-gate frame choice can place any individual target in the
well-conditioned region. We leave a systematic treatment to future
work. 

\begin{figure}[!t]
\centering
\includegraphics[width=\columnwidth]{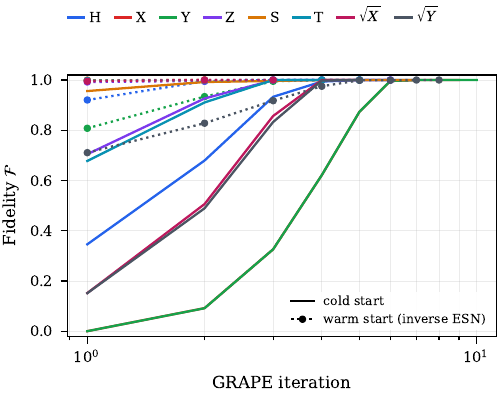}
\caption{Warm- versus cold-start convergence of the accelerated
\texttt{GRAPE} kernel on the standard eight-gate suite: gate fidelity
versus iteration (logarithmic axis), with solid lines for the cold start
and dotted lines with markers for the warm start supplied by the inverse
ESN; every run stops at $\mathcal F>0.99999$. Both starts converge
within ten iterations at one qubit; the warm start enters at its
single-pass fidelity (Table~\ref{tab:hybrid}) and lowers the median
iteration count from $6.5$ to $5$. The corresponding comparison for
plain gradient ascent (last two columns of Table~\ref{tab:hybrid}) shows
cold-start convergence at $700$--$1800$ iterations, roughly halved by
the warm start on seven of the eight gates.}
\label{fig:hybrid}
\end{figure}

\section{Discussion}
\label{sec:discussion}

In this work, we introduced a structure-preserving framework for quantum gate synthesis that combines the Wei--Norman decomposition with reservoir computing. Rather than learning the propagator directly, the proposed method learns the evolution of a finite set of real-valued coordinates, from which the propagator is reconstructed analytically. This separation gives the two components of the framework distinct and complementary roles: the reservoir captures the nonlinear, history-dependent response of the system to the control fields, while the Wei--Norman reconstruction enforces the underlying Lie-group structure exactly. Consequently, errors in the predicted coordinates may reduce the resulting gate fidelity, but they cannot produce a nonunitary evolution.

This property is more than a formal convenience. Direct machine-learning models of quantum dynamics must generally learn both the system evolution and the constraints defining the physically admissible output space. In the present construction, the physical validity is imposed by design rather than inferred from data. Moreover, the theoretical bounds derived above make this division of labor explicit: under standard regularity assumptions, the error in the reconstructed propagator is controlled by the approximation error of the coordinate surrogate.

The present study also opens several directions for future research, the most pressing of which is extending the demonstration from the single qubit to gates of higher order. The Wei--Norman chart for $q$ qubits involves $n=4^q-1$ coordinates, so already the two- and three-qubit cases---the natural next targets, encompassing entangling gates such as \texttt{CNOT}, \texttt{CZ}, and \texttt{Toffoli}---bring the coordinate dimension from $3$ to $15$ and $63$. The existence of a convenient global factorization is nontrivial for these higher-rank algebras; a realistic route combines subsystem decompositions of Cartan (KAK) type with Wei--Norman charts on the factors, an avenue we intend to pursue directly. This is also the regime in which established optimizers become genuinely expensive, and in which a learned warm start should outperform a cold-started optimizer by a measurable margin in time-to-threshold, widening the factor-of-two advantage observed against plain gradient ascent at one qubit in Sec.~\ref{sec:results}. As the number of coordinates grows, the question of whether physical structure---sparse interactions, symmetries, or a hierarchical gate decomposition---can be exploited to reduce the effective dimension of both the coordinate representation and the reservoir becomes increasingly important, and we expect the two- and three-qubit cases to be the setting in which this question is first answered concretely.

\section*{Data availability}
The code and data reproducing all experiments are available at
\url{https://doi.org/10.5281/zenodo.21377857}. 

\section*{Acknowledgments}
FC thanks F.~Anzà for useful discussions and introducing the authors to the Wei--Norman formula. All the authors thank D. De Santis for comments on an initial version of the draft.
FC's work was conducted under the auspices of the National Nuclear Security Administration of the United States Department of Energy at Los Alamos National Laboratory (LANL) under Contract No. DE-AC52-06NA25396. FC is currently an employee of Planckian srl, but acknowledges funding from LDRD (project 20240245ER).
RM and AS acknowledge partial support from the Quantum Summer School Program 2025 at LANL's Theoretical Division. 
AS acknowledges the Spanish State Research Agency, through the Mar\'ia de Maeztu project CEX2021-001164-M funded by the MICIU/AEI/10.13039/501100011033, through the COQUSY project PID2022-140506NB-C21 and -C22 funded by MICIU/AEI/10.13039/501100011033, MINECO through the QUANTUM SPAIN project, and EU through the RTRP - NextGenerationEU within the framework of the Digital Spain 2025 Agenda. A.S. also acknowledges the CSIC Interdisciplinary Thematic Platform (PTI+) on Quantum Technologies in Spain (QTEP+) and the support of a fellowship from the ``la Caixa” Foundation (ID 100010434 - LCF/BQ/DI23/11990081). A. S. also acknowledges support from the U.S. Department of Energy (DOE) through a quantum computing program sponsored by the Los Alamos National Laboratory (LANL) Information Science $\&$ Technology Institute.
AS has also been supported by the USRA Feynman Quantum Academy internship program. 

\appendix 
\section{Hyperparameter sweep data}
\label{app:sweep}

This appendix collects the hyperparameter sweeps summarized in
Sec.~\ref{sec:forward}. All sweeps are on $SO(3)$ with $150$ training
trajectories and the baseline configuration $R=1000$, $\rho=0.95$,
$\alpha=0.1$, $\sigma_{\mathrm{in}}=1.0$, $\lambda=10^{-8}$. Table~\ref{tab:sweeps} reports the one-dimensional scans; the input scaling,
the most sensitive parameter, is shown in Fig.~\ref{fig:sweep_is}.

\begin{table}[h]
\centering
\caption{One-dimensional hyperparameter sweeps.}
\label{tab:sweeps}
\smallskip
\begin{tabular}{llcc}
\toprule
Parameter & Value & Mean $\|\bm{\xi}-\hat{\bm{\xi}}\|$
 & Mean $\mathcal{F}$ \\
\midrule
\multirow{5}{*}{$R$}
 & 100 & $1.63\times10^{-1}$ & 0.994 \\
 & 200 & $2.10\times10^{-1}$ & 0.989 \\
 & 500 & $2.41\times10^{-1}$ & 0.983 \\
 & 1000 & $4.96\times10^{-1}$ & 0.929 \\
 & 2000 & $7.22\times10^{-1}$ & 0.840 \\
\midrule
\multirow{4}{*}{$\alpha$}
 & 0.01 & $1.15\times10^{-1}$ & 0.994 \\
 & 0.05 & $3.34\times10^{-1}$ & 0.951 \\
 & 0.1 & $4.96\times10^{-1}$ & 0.929 \\
 & 0.3 & $9.19\times10^{-1}$ & 0.773 \\
\midrule
\multirow{4}{*}{$\rho$}
 & 0.5 & $4.14\times10^{-1}$ & 0.953 \\
 & 0.8 & $3.31\times10^{-1}$ & 0.970 \\
 & 0.95 & $4.96\times10^{-1}$ & 0.929 \\
 & 0.99 & $4.34\times10^{-1}$ & 0.942 \\
\midrule
\multirow{4}{*}{$\sigma_{\mathrm{in}}$}
 & 0.1 & $6.27\times10^{-2}$ & 0.999 \\
 & 0.5 & $2.13\times10^{-1}$ & 0.978 \\
 & 1.0 & $4.96\times10^{-1}$ & 0.929 \\
 & 3.0 & $9.95\times10^{-1}$ & 0.833 \\
\midrule
\multirow{4}{*}{$N_{\mathrm{train}}$}
 & 10 & $4.84\times10^{-1}$ & 0.955 \\
 & 50 & $6.12\times10^{-1}$ & 0.862 \\
 & 100 & $1.81\times10^{-1}$ & 0.987 \\
 & 200 & $1.53\times10^{-1}$ & 0.989 \\
\bottomrule
\end{tabular}
\end{table}

\begin{figure}[!t]
\centering
\includegraphics[width=\columnwidth]{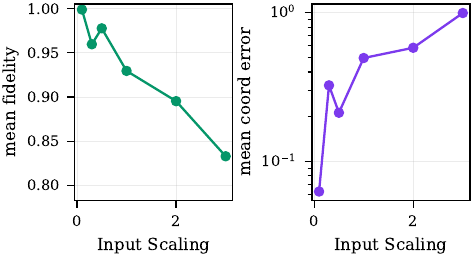}
\caption{Fidelity and coordinate error versus input scaling
$\sigma_{\mathrm{in}}$, the most sensitive hyperparameter.}
\label{fig:sweep_is}
\end{figure}

\section{Extended statistics for the inverse ESN}
\label{app:invstats}

This appendix records the quantitative output of the inverse-ESN
experiment of Sec.~\ref{sec:synthesis} beyond the summary in the main
text.

\emph{Training-set coverage.} The $8000$ pulse--gate pairs were
generated in $61$\,s. The coverage in Wei--Norman coordinates is
broad: $\xi_1\in[-8.35,5.80]$ (s.d.\ $2.09$), $\xi_2\in[0.15,2.99]$
(s.d.\ $0.65$), $\xi_3\in[-8.45,8.12]$ (s.d.\ $2.97$). The nearest generated
example to the canonical representative of the Hadamard target lies at
Euclidean distance $0.2309$ in coordinate space, so the training set
samples the neighborhood of the principal benchmark gate without
memorizing it. The edges $\xi_2=0.15$ and $\xi_2=\pi-0.15$ are not
incidental: the coordinate is clipped to $[0.15,\pi-0.15]$ during
generation, so the training set avoids the singular set by
construction, and the single-pass losses on the frame-displaced
boundary targets (the Hadamard and $\sqrt{Y}$ in the main text) combine
the chart ill-conditioning with this deliberate lack of coverage. 

\emph{Regression quality.} Training the readout on the $8000$ samples
yields a Fourier-coefficient RMSE of $2.4423\times10^{-1}$ (per-pair
median $1.8735\times10^{-1}$, maximum $4.9454\times10^{-1}$), in a
single ridge-regression solve; for scale, the coefficients are sampled
with standard deviation $\mathrm{amp}/\sqrt{n_{\mathrm{modes}}+1}=0.6$
before mode tapering. More informative is the induced gate quality
after reconstruction and integration: on $50$ randomly selected training
pairs the mean fidelity is $0.9844$ (median $0.9929$, minimum $0.8498$).
Coefficient-level error is thus compatible with high gate fidelity after
reconstruction.

\emph{Held-out generalization.} Table~\ref{tab:app_general_test}
expands the held-out statistics quoted in the main text. The
distribution is strongly skewed toward high fidelity; the minimum
fidelity $3.2\times10^{-3}$ shows that catastrophic failures remain
possible but rare compared with the dominant high-fidelity bulk
(Fig.~\ref{fig:inv_fids}).

\begin{table}[h]
\centering
\caption{Held-out performance of the inverse ESN \\ on $300$ unseen
targets.}
\label{tab:app_general_test}
\smallskip
\begin{tabular}{lc}
\toprule
Statistic & Value \\
\midrule
Mean fidelity & 0.943072 \\
Median fidelity & 0.993538 \\
Minimum fidelity & 0.003218 \\
Maximum fidelity & 0.999892 \\
Mean infidelity & $5.6928\times10^{-2}$ \\
Fraction with $\mathcal{F}>0.99$ & 63.0\% \\
Fraction with $\mathcal{F}>0.95$ & 83.0\% \\
Fraction with $\mathcal{F}>0.90$ & 90.7\% \\
Fraction with $\mathcal{F}>0.80$ & 93.7\% \\
\bottomrule
\end{tabular}
\end{table}

\bibliographystyle{IEEEtran}
\bibliography{biblio}
\end{document}